\documentclass[11pt,a4paper]{article}
\usepackage[hidelinks]{hyperref} 
\usepackage{amsthm}
\usepackage{thm-restate}
 \usepackage{fullpage}
\usepackage[utf8x]{inputenc}
\usepackage{todonotes,wrapfig,float,graphicx,amssymb,textcomp,array,amsmath}
\usepackage{enumerate,enumitem}
\usepackage{subcaption}
\usepackage{multirow}
\usepackage{lineno}
\usepackage{tabularx}
\usepackage{color,xcolor}

\usepackage{lineno}
\usepackage{cleveref}
\newcommand{\old}[1]{{}}
\newcommand{\later}[1]{{}}
\def\defn#1{\textit{\textbf{\boldmath #1}}}
\renewcommand{\emph}[1]{\defn{#1}} 

\newtheorem{theorem}{Theorem}
\newtheorem{lemma}[theorem]{Lemma}

\newtheorem{clm}[theorem]{Claim}
\newtheorem{corollary}[theorem]{Corollary}

\newtheorem{observation}[theorem]{Observation}
\newtheorem{rmk}[theorem]{Remark}

\DeclareMathOperator{\dist}{dist}

\DeclareMathOperator{\Vol}{Vol}

\newcommand{\ACY}{\textsf{\sc Classify}}

\newcommand{\RCY}{\textsf{\sc HR-Classify}}

\newcommand{\AF}{\textsf{\sc Filter}}

\newcommand{\AGDS}{\textsf{\sc FirstFit}}

\newcommand{\IR}{\mathbb{R}}

\def\OO{{\cal O}} 

\def\A{{\cal A}}

\def\I{{\mathcal I}}

\def\P1{{\mathcal P_1}}
\def\P2{{\mathcal P_2}}
\def\P{{\mathcal P}}

\def\C{{\mathcal C}}
\usepackage{lineno}

\definecolor{Darkblue}{rgb}{0,0,.8}
\definecolor{Brown}{cmyk}{0,0.61,1.,0.60}
\definecolor{Purple}{cmyk}{0.45,0.86,0,0}
\definecolor{Darkgreen}{rgb}{0.133,0.700,0.133}

\definecolor{MyGreen}{rgb}{0.30,0.400,0.30}
\renewcommand{\emph}[1]{{\color{Brown!85}{\em #1}}}

\usepackage{color}

\newcommand{\myqed}{\hfill $\ $}

\begin{document}
\title{Online Algorithms for Geometric Independent Set}
\author{Minati De\thanks{Dept. of Mathematics, Indian Institute of Technology Delhi, New Delhi, India. Email: \texttt{Minati.De@maths.iitd.ac.in}.}
\and
Satyam Singh\thanks{Department of Computer Science, Aalto University, Espoo, Finland. Email: \texttt{satyam.singh@aalto.fi}. Research on this paper was supported by the Research Council of Finland, Grant 363444. }}
\date{}

\maketitle             
\thispagestyle{empty}

\begin{abstract} 
In the classical online model, the maximum independent set problem admits an $\Omega(n)$ lower bound on the competitive ratio even for interval graphs, motivating the study of the problem under additional assumptions.
We first study the problem on graphs with a bounded independent kissing number $\zeta$, defined as the size of the largest induced star in the graph minus one.  We show that a simple greedy algorithm, requiring no geometric representation, achieves a competitive ratio of $\zeta$. Moreover, this bound is optimal for deterministic online algorithms and asymptotically optimal for randomized ones. This extends previous results from specific geometric graph families to more general graph classes.

Since this bound rules out further improvements through randomization alone, we investigate the power of randomization with access to geometric representation.
When the geometric representation of the objects is known, we present randomized online algorithms with improved guarantees. For unit ball graphs in $\IR^3$, we present an algorithm whose expected competitive ratio is strictly smaller than the deterministic lower bound implied by the independent kissing number.
For $\alpha$-fat objects and for axis-aligned hyper-rectangles in $\mathbb{R}^d$ with bounded diameters, we obtain algorithms with expected competitive ratios that depend polylogarithmically on the ratio between the maximum and minimum object diameters. In both cases, the randomized lower bound implied by the independent kissing number grows polynomially with the ratio between the maximum and minimum object diameters, implying substantial performance guarantees for our algorithms.
\end{abstract}

\newpage
\section{Introduction}
\pagenumbering{arabic}
The maximum independent set problem is a fundamental problem in combinatorial optimization~\cite{AdamaszekHW19,ChalermsookC09,ChanH12,ChristodoulouZ05,CormodeDK19,ErlebachJS05,GareyJ79,Karp,Mitchell21}.
Given a graph $G=(V,E)$, an \emph{independent set} (IS) is a subset $I\subseteq V$
such that no two vertices in $I$ are adjacent.
The \emph{maximum independent set} (MIS) problem seeks an independent set of maximum cardinality.

The MIS problem has been extensively studied across a wide range of computational models,
reflecting its importance in both theory and applications~\cite{GareyJ79,OliverHKSV14,HalldorssonIMT02,Karp,MaratheBHRR95,Vazirani01}.
One particularly challenging setting is the \emph{classical online model of computation}, where the vertices of a graph arrive one by one, and an algorithm must maintain a feasible independent set by irrevocably deciding, upon arrival, whether to accept or reject each vertex. The \emph{competitive ratio} (see~\Cref{Sect_Notation} for a formal definition) is the \emph{primary} performance measure of online algorithms.

In the classical online model, the online MIS problem is notoriously difficult~\cite{CaragiannisFKP07,ErlebachF06,OliverHKSV14,LiptonT94,MaratheBHRR95}.
It is well known that no online algorithm can achieve a competitive ratio better than $\Omega(n)$ even for interval graphs, where $n$ is the number of intervals~\cite{LiptonT94}. More generally, for arbitrary graphs in the classical online setting, such a pessimistic lower bound rules out the possibility of any sublinear competitive ratio.
As a consequence, a meaningful guarantee cannot be obtained without imposing an additional structure on the input.

This has led to two main lines of research.
One direction, which we adopt in this paper, is to focus on restricted subclasses of geometric intersection graphs\footnote{A \emph{geometric intersection graph} is a graph where vertices correspond to geometric objects and two vertices are adjacent if and only if the corresponding objects intersect.}, where geometric and combinatorial structure can be exploited by online algorithms~\cite{CaragiannisFKP07,ErlebachF06,MaratheBHRR95}.
Another direction addresses this difficulty by relaxing the classical online model, for example, by considering models that relax the irrevocability requirement~\cite{BoyarFKL22}, stochastic model~\cite{OliverHKSV14}, random-order arrival model~\cite{GargKK24}, or models permitting recourse~\cite{BergSS25}.
These models offer stronger performance guarantees but are not directly comparable to results from the classical online setting. We now briefly review the relevant prior work on the online MIS problem in the classical online model.

\subsection{Related Work}
Due to the pessimistic lower bound of $\Omega(n)$ known for interval graphs, there is a scarcity of literature on the online MIS problem in the classical online model.
For unit disk graphs, Marathe et al.~\cite{MaratheBHRR95} showed that the greedy algorithm $\AGDS$, which does not require disk representation, achieves a competitive ratio of at most~5.
 Later, Erlebach and Fiala~\cite{ErlebachF06} showed that this bound is tight, which proved that no deterministic online algorithm can achieve a better competitive ratio even when the geometric representation is available.
They also studied disk graphs having radii in the range $[1,M]$ and showed that $\AGDS$ achieves a competitive ratio of~$O(M^2)$.
Caragiannis et al.~\cite{CaragiannisFKP07} complemented this result by showing that, when the disk representation is not available, any deterministic or randomized online algorithm on disk graphs with radii in $[1,M]$ has a competitive ratio ~{$\Omega(M^2)$} under oblivious adversaries\footnote{In \emph{oblivious adversary}, the adversary fixes the entire input sequence in advance, knowing the algorithm but not the outcomes of its random choices. 
The performance of a randomized algorithm is evaluated by averaging over its random choices on the fixed input sequence.}.
These results provide a fairly complete picture of disk graphs with bounded radii. Beyond these restricted subclasses, the performance of the greedy algorithm $\AGDS$ remains largely unexplored, which naturally leads to the following question.

\begin{itemize}
\item[Q1.]  Can the performance guarantee of the greedy algorithm $\AGDS$ be extended to a broader class of graphs, and is this guarantee optimal even for any randomized algorithms?
\end{itemize}

Caragiannis et al.~\cite{CaragiannisFKP07} demonstrated that randomization together with the geometric representation can surpass the deterministic lower bound of $5$ for unit disk graphs. In particular, they presented a randomized algorithm with an expected competitive ratio of~4.41, against oblivious adversaries. They also studied disk graphs with radii in $[1,M]$ and presented a randomized algorithm achieving an asymptotically optimal expected competitive ratio of $\Theta(\log M)$, thereby surpassing the corresponding deterministic lower bound of $\Omega(M^2)$.
To the best of our knowledge, in the classical online model, no positive results are known apart from disk graphs with bounded radii. This raises a second, broader question.

\begin{itemize} 
\item[Q2.] Can randomization and geometric structure be systematically exploited to obtain positive results for more general families of geometric objects, such as similarly sized fat objects or axis-aligned hyper-rectangles in $\IR^d$? 
\end{itemize}

\subsection{Our Contributions}\label{sec_contri}
In this paper, we provide positive answers to both of the above-mentioned questions, as summarized below.
\paragraph*{Analysis of the greedy algorithm beyond disk graphs.}
We analyze the natural greedy algorithm $\AGDS$ for the online MIS problem on graphs with bounded independent kissing number (see~\Cref{Sect_Notation} for the definition), and show that this parameter exactly characterizes the optimal performance up to a constant factor achievable by any online algorithm.
\begin{restatable}{theorem}{thmGreedy}\label{thm:ind}
The algorithm $\AGDS$ achieves a competitive ratio of $\zeta$ for the online MIS problem on any graph whose independent kissing number is at most $\zeta$.
Moreover, this bound is optimal for deterministic algorithms: no deterministic online algorithm can achieve a competitive ratio smaller than $\zeta$ against adaptive adversaries on this class of graphs. 
For randomized algorithms, the bound is asymptotically optimal: no randomized online algorithm can achieve a competitive ratio better than $(\zeta+1)/2$
against oblivious adversaries on this class of graphs.
\end{restatable}
This broadens earlier analyses~\cite{CaragiannisFKP07,ErlebachF06,MaratheBHRR95} for specific geometric graph families, such as unit disk graphs and disk graphs with radii in $[1,M]$, to more general graph classes.

\paragraph*{Use of randomization along with geometric representation.}
The lower bound by~\Cref{thm:ind} ensures that randomization alone would not be much help. Therefore, we investigate whether randomization, together with access to geometric representations, can improve competitive guarantees for the online MIS problem, surpassing the lower bound obtained in~\Cref{thm:ind}.

\smallskip\noindent
\textit{\underline{Unit balls in $\IR^3$.}}
We extend the two-dimensional randomized algorithm for unit disks  of Caragiannis et al.~\cite{CaragiannisFKP07} to unit balls in $\IR^3$. Our approach is based on a new lattice (extendable to any higher dimension) with non-trivial key structural properties and yields the following provable guarantees. 
\begin{restatable}{theorem}{thmUnitDisk}\label{thm:rand-alpha}
There exists a randomized online algorithm for the online MIS problem on the geometric intersection graph of unit balls in $\IR^3$ that achieves an expected competitive ratio of at most $\frac{1}{\pi}(36+\varepsilon)\equiv 11.46$ against oblivious adversaries, where $\varepsilon>0$ is an arbitrarily small constant.
\end{restatable}

Since the independent kissing number is exactly~12 for unit ball graphs in $\IR^3$~\cite[Thm.~7 (a)]{DeKS26}, the competitive ratio obtained by~\Cref{thm:rand-alpha} surpasses the deterministic lower bound of~12 obtained in~\Cref{thm:ind}.
Our algorithm and its analysis extend to higher dimensions; however, the resulting competitive ratios do not surpass the deterministic lower bounds for dimensions $d \ge 4$; see~\Cref{sect_extension} for details.


\smallskip\noindent
\textit{\underline{Similarly Sized Fat Objects in $\IR^d$.}}  
    For $\alpha$-fat objects in $\IR^d$ whose widths\footnote{Informally, the \emph{width} of an object is defined as the radius of the largest ball enclosed in the object.} lie in the range $[1,M]$, with $M>2$, we present a randomized online algorithm with provable competitive ratio guarantees. 
\begin{restatable}{theorem}{thmFat}\label{thm:ACY}
   Let $M>2$, and let $\zeta'$ be the independent kissing number of $\alpha$-fat objects with widths in $[1,2]$. There exists a randomized online algorithm for the online MIS problem on the geometric intersection graph of $\alpha$-fat objects in $\IR^d$ with widths in the range $[1,M]$ that achieves an expected competitive ratio of at most~{${\zeta'}({\lfloor\log M\rfloor}+1)$} against oblivious adversaries.
\end{restatable}

To obtain this, we use a standard decomposition by partitioning the objects into $\lfloor \log M \rfloor + 1$ classes according to their widths. Selecting one class uniformly at random and running the greedy algorithm $\AGDS$ on that class yields an expected competitive ratio of $O(\zeta' \log M)$, where $\zeta'$ denotes the independent kissing number for objects with widths in $[1,2]$. 
By~\cite[Thm.~7 (f)]{DeKS26}, we have $\zeta' \le \left(\frac{2}{\alpha}+2\right)^d$, while the independent kissing number $\zeta$ for objects with widths in $[1,M]$ satisfies $\zeta \ge \left(\frac{\alpha}{2}\frac{M+2}{1+\varepsilon}\right)^d$ for any constant $\varepsilon>0$. Hence, $\zeta = \Omega((\alpha M)^d)\zeta'$. Therefore, the competitive ratio achieved by our randomized algorithm is significantly smaller than the deterministic lower bound~$\zeta$ (by \Cref{thm:ind}).

\smallskip\noindent
    \textit{\underline{Similarly Sized Hyper-rectangles in $\IR^d$.}}  
Axis-aligned hyper-rectangles in $\IR^d$ with side lengths in $[1,M]$ form a special subclass of $\alpha$-fat objects with $\alpha=O(1/(M\sqrt{d}))$.  Now, using the known bound $\zeta'=O((2/\alpha)^d)$ due to~\cite[Thm.~7 (f)]{DeKS26}, the aforementioned~\Cref{thm:ACY} implies an expected competitive ratio of $O((M\sqrt{d})^d\log M)$ for this class.
    By exploiting the additional structure of hyperrectangles, we further obtain significantly improved performance.
\begin{restatable}{theorem}{thmRect}\label{thm:RCY}
    For any $M>2$, there exists a randomized online algorithm for the online MIS problem on the geometric intersection graph of hyper-rectangles in $\IR^d$ having side lengths in $[1,M]$ that achieves an expected competitive ratio of at most~$(4\left({\lfloor\log M\rfloor}+1\right))^d$ against oblivious adversaries.
\end{restatable}
The main challenge in proving \Cref{thm:RCY} is to control the independent kissing number after refining the input into $(\log M)^d$ subclasses based on side lengths. Although this refinement reduces the dependence on $M$, it is not immediate that the objects within a single subclass admit a constant independent kissing number. Establishing that this number is bounded solely as a function of the dimension $d$, and independent of $M$, is crucial for obtaining the stated competitive ratio.

\subsection{Organization.}
In~\Cref{Sect_Notation}, we introduce notation, definitions, and preliminaries.
In~\Cref{sec_deter}, we analyze the greedy algorithm and prove optimal deterministic and randomized lower bounds for the online MIS problem for graphs with bounded independent kissing number.
\Cref{sec_disk} presents a lattice-based randomized algorithm for unit ball graphs in $\mathbb{R}^3$ that improves upon the deterministic bound.
In~\Cref{sec_fat}, we develop a randomized approach for the online MIS problem for similarly sized $\alpha$-fat objects in $\mathbb{R}^d$.
In~\Cref{sec_rect}, for axis-aligned hyper-rectangles in $\mathbb{R}^d$,   exploiting their structure, we present a randomized algorithm obtaining stronger guarantees.
Finally,~\Cref{sect_conclusion} concludes with a discussion on open problems.

\section{Notation and Preliminaries}\label{Sect_Notation}
For a positive integer $n$, we use $[n]$ to represent the set $\{1,2,\ldots,n\}$. For a point $p\in\IR^d$ and an index $i \in [d]$, we use $p(x_i)$ to denote the $i$th coordinate of $p$. For any two points $a,b\in \IR^d$, we denote by $\dist(a,b)$ the Euclidean distance between $a$ and $b$. By an \emph{object}, we refer to a compact set in $\mathbb{R}^d$ with a nonempty interior. Two objects $u$ and $v$ are said to \emph{intersect} each other if $u\cap v\neq \emptyset$.
They are called \emph{non-overlapping} if their interiors are disjoint, and they are \emph{non-touching} if $u \cap v = \emptyset$. Note that two non-overlapping objects may still intersect at their boundaries. 
Given a family $\cal S$ of geometric objects in $\IR^d$, the \emph{geometric intersection graph} $G$ of $\cal S$ is an undirected graph $G(V,E)$, where $V=\cal S$, and $E=\{\{u,v\}| u,v\in {\cal S} \text{ and } u\cap v\neq \emptyset \}$.
 For an object $\sigma \subseteq \IR^3$, we denote its volume by $\Vol(\sigma)$. A \emph{unit ball} is a ball of radius~$1$.

\smallskip\noindent
\textbf{Fat Object.}
Let $\sigma$ be an object and let $x\in \sigma$ be any point. Let $\alpha(x)$ denote the ratio between the minimum and maximum Euclidean distance from $x$ to the boundary $\delta(\sigma)$ of $\sigma$. 
	The \emph{aspect ratio} of $\sigma$ is defined as  $\alpha=\max \{\alpha(x)  \mid x \in  \sigma\}$. We call $\sigma$ an \emph{$\alpha$-fat object} if its aspect ratio is at least $\alpha$.
A point $x \in \sigma$ at which the maximum is attained is called an
\emph{aspect point} of $\sigma$. Such a point need not be unique. The minimum and maximum Euclidean distances from an aspect point to the boundary $\delta(\sigma)$ are referred to as the \emph{width} and \emph{height} of $\sigma$, respectively. 
  For an $\alpha$-fat object, the value of $\alpha$ is invariant under translation, rotation, reflection, and scaling. 
 For further details, we refer to~\cite{DeJKS24}.

A family $\cal S$ of objects is called \emph{fat} if there exists a constant  $\alpha\in(0,1]$ such that every object in $\cal S$ is $\alpha$-fat. Note that objects in $\cal S$ need not be convex or connected.
We say that $\cal S$ consists of \emph{similarly sized fat objects} if the ratio of the largest width to the smallest width among the objects in $\cal S$ is bounded by a fixed constant.

\smallskip\noindent
\textbf{Independent Kissing Number.}
Let $G=(V, E)$ be a graph. For any subset $V'\subseteq V$, the \emph{induced subgraph} $G[V']$ is the graph whose vertex set is $V'$ and whose edge set consists of all of the edges in $E$ with both endpoints in  $V'$. We denote by $\varphi(G)$ the size of a maximum independent set (MIS) of $G$. For a vertex $v\in V$,  let $N(v)=\{u(\neq v) \in V | \{u,v\} \in E\}$ denote the neighborhood of the vertex $v$. The \emph{independent kissing number} of a graph  $G=(V,E)$  is defined as $\zeta=\max_{v \in V} \{\varphi(G[N(v)])\}$.
 Equivalently, a graph has independent kissing number $\zeta$ if and only if it is $K_{1,\zeta+1}$-free. 
Note that $\zeta$ can be much smaller than the number of vertices in the graph.
For  a detailed study on the independent kissing number, we refer to~\cite{DeKS26}.

\smallskip\noindent
\textbf{Competitive Ratio.} Let $\A$ be an online algorithm for a maximization problem.  We say that $\A$ is \emph{$c$-competitive} if $c=\sup_{\cal I}\frac{\OO_{\cal I}}{\A_{\cal I}}$,  where $\A_{\cal I}$ and $\OO_{\cal I}$ denote the values of the solutions produced by $\A$ and by an optimal offline algorithm, respectively, on an input sequence $\cal I$.  If $\A$ is a randomized algorithm, then $\A_{\cal I}$ is replaced by its expected value~\cite{BorodinE,BorodinP}.


\section{Optimal Competitive Deterministic Algorithm }\label{sec_deter}
In this section, we prove~\Cref{thm:ind}.
The proof of~\Cref{thm:ind} is based on three lemmas.
\Cref{lem_DUB} establishes an upper bound on the competitive ratio of the greedy algorithm  $\AGDS$ in terms of the independent kissing number. \Cref{lem_DLB} presents a matching lower bound for all deterministic online algorithms, while~\Cref{lem_RLB} presents an 
asymptotically matching lower bound for all randomized online algorithms, obtained by generalizing the construction of Caragiannis et al.~\cite[Thm.~7]{CaragiannisFKP07} with slight effort.
Together, these lemmas imply~\Cref{thm:ind}: $\AGDS$ is an optimal competitive deterministic algorithm for the problem.

The algorithm $\AGDS$ works as follows: It maintains an independent set $\A$. Initially, $\A=\emptyset$. 
Upon the arrival of a vertex $v$, the algorithm adds $v$ to $\A$ if and only if $v$ is not adjacent to any vertex already in $\A$.

\begin{restatable}{lemma}{lemDUB}\label{lem_DUB}
    The algorithm $\AGDS$ achieves a competitive ratio of $\zeta$ for the online MIS problem on any graph whose independent kissing number is at most~$\zeta$ under adaptive adversaries\footnote{In \emph{adaptive adversary}, the adversary generates the input sequence in online fashion, choosing each request as a function of the algorithm’s past actions.}. 
\end{restatable}
\begin{proof}
     Let $G(V,E)$ be the graph revealed to the algorithm, and let $\zeta$ denote the independent kissing number of $G$. Let $\A$ and $\OO$  be the sets of vertices reported by our algorithm $\AGDS$ and by an optimal offline algorithm, respectively. 
     Since $\AGDS$ constructs an independent set, every vertex not selected by the algorithm must be adjacent to some selected vertex; hence $\A$ is a dominating set\footnote{For a graph $G=(V,E)$, a subset $D \subseteq V$ is called a \emph{dominating set} if for every vertex $v \in V$, either $v \in D$ or there exists an edge $\{u,v\} \in E$ with $u \in D$.} of $G$.
Let $\C=\A\cap \OO$, and let $\A'=\A \setminus \C$ and $\OO'=\OO\setminus \C$. Consider any vertex $a\in \A'$, and let $\OO_a\subseteq \OO'$ denote the set of vertices in $\OO'$ that are adjacent to $a$. Since $\OO_a$ is an independent set and $G$  has the independent kissing number $\zeta$, we have  $|\OO_a|\leq \zeta$.
Now consider any vertex $o\in \OO'$. Since $\A$ is a dominating set of $G$ and $o\notin \A$, the vertex $o$ must be adjacent to some vertex in $\A$. Moreover, since $\OO=\C\cup \OO'$ is an independent set, $o$ cannot be adjacent to any vertex in $\C$. Therefore, $o$ must be adjacent to some vertex $a\in \A'(=\A\setminus \C)$. Equivalently, $o\in \OO_a$ for some $a\in \A'$.
 It follows that, $\OO'=\cup_{a\in \A'}\OO_a$. Consequently, $|\OO'|\leq \zeta|\A'|$, resulting in $|\OO|\leq \zeta|\A|$.  Therefore, the algorithm $\AGDS$ achieves a competitive ratio of at most~$\zeta$.
\end{proof}

\begin{restatable}{lemma}{lemDLB}\label{lem_DLB}
No deterministic online algorithm can achieve a competitive ratio better than~$\zeta$ for the online MIS on graphs whose independent kissing number is at most~$\zeta$ {under adaptive adversaries}.
\end{restatable}
\begin{proof}
We describe an adaptive adversarial strategy against any deterministic online algorithm. The adversary constructs a graph $G$ online by revealing vertices one by one, together with their adjacencies to previously revealed vertices.

The adversary begins by revealing a single vertex $v_0$.
If the online algorithm rejects $v_0$, then no further vertices are revealed. In this case, the algorithm outputs the empty set, whereas an optimal offline algorithm selects $v_0$, resulting in an unbounded competitive ratio.
Hence, we may assume that the algorithm accepts $v_0$.

The adversary then reveals $\zeta$ additional vertices $v_1,\ldots,v_\zeta$ such that these vertices are pairwise non-adjacent and each of them is adjacent to $v_0$. Since all these vertices are adjacent to $v_0$, the algorithm cannot accept any of them. Thus, the algorithm outputs the independent set $\{v_0\}$. On the other hand, the vertices $v_1,\ldots,v_\zeta$ form an independent set.
Therefore, an optimal offline algorithm can select all $\zeta$ vertices.
The resulting competitive ratio of the algorithm is at least~$\zeta$.
This completes the proof.
\end{proof}

Generalizing the lower bound construction of Caragiannis et al.~\cite[Thm.~7]{CaragiannisFKP07} that studies the online MIS problem for disks with bounded radii (without disk representation), one can obtain the following.
\begin{restatable}{lemma}{lemLower}\label{lem_RLB}
    No randomized online algorithm can achieve a competitive ratio better than~$(\zeta+1)/2$ for the online MIS problem on graphs whose independent kissing number is at most~$\zeta$, under oblivious adversaries. 
\end{restatable}
\begin{proof}
Here the adversarial construction follows~\cite[Thm.~7]{CaragiannisFKP07}, though our verification that $G_\zeta$ has independent kissing number at most $\zeta$ (Claim~\ref{clm_zeta}) replaces their geometric argument with a short combinatorial one, generalizing beyond unit disk graphs.

We construct an oblivious adversary that generates a graph $G_\zeta$  with the independent kissing number at most~$\zeta$. We then show that $G_\zeta$ contains an IS of size $\zeta+1$, while the expected size of the IS produced by any randomized online algorithm is at most 2.

The graph $G_\zeta$ is generated online and consists of $\zeta$ levels. For each $i\in[\zeta]$, the level~$i$ contains two vertices, denoted by $v_i^{\ell}$ and $v_i^r$, which are not adjacent to each other. The adversary first reveals the two vertices of level $1$. For each $i\in[\zeta]\setminus\{1\}$, the vertices of level $i$ are revealed only after all vertices of level $i-1$ have been revealed.

For each level $i\in[\zeta]\setminus\{1\}$, the adversary tosses a fair coin to determine how the vertices of level $i$ are connected to the vertices of previous levels. If head comes up, both vertices $v_i^{\ell}$ and $v_i^r$ are connected to vertex $v_{i-1}^{\ell}$ and to all vertices of levels $1,\ldots,i-2$ to which $v_{i-1}^{\ell}$ is connected. Otherwise, both vertices $v_i^{\ell}$ and $v_i^r$ are connected to vertex $v_{i-1}^r$ and to all vertices of levels $1,\ldots,i-2$ to which $v_{i-1}^r$ is connected.
Consider the set consisting of the two vertices of level $\zeta$ and, for each $i\in[\zeta-1]$, the unique vertex of level $i$ that is not connected to any vertex of a higher level. By construction, this set is an IS of $G_\zeta$. Hence, the size of a MIS of $G_\zeta$ has size at least $\zeta+1$.
    
We now bound the expected size of an independent set computed by any online algorithm on $G_\zeta$. By Yao’s minimax principle~\cite{Yao77}, it suffices to consider deterministic online algorithms that know the probability distribution used by the adversary. Fix such a deterministic algorithm. The algorithm can accept at most one vertex from each level, except possibly from the last level, where it may accept both vertices. For each level $i$, let $X_i$ be the indicator random variable that equals $1$ if the algorithm accepts a vertex from level $i$, and $0$ otherwise.

We now analyze the probability that the algorithm successfully accepts a vertex from a given level.
For any level $i$, if the algorithm accepts a vertex at level $i$, this vertex remains feasible only if, for every earlier level $j<i$, the adversary’s coin toss at level $j$ selects the branch that does not connect level $j+1$ to this vertex. Since the coin
tosses at different levels are independent and each outcome occurs with probability $1/2$, the probability that a vertex accepted at level $i$ survives all earlier levels is at most $2^{-(i-1)}$.
Therefore, we have
\[
\Pr[X_i = 1] \le 2^{-(i-1)}.
\]

Summing over all levels, the algorithm can accept at most one vertex from each level $i<\zeta$ and at most two vertices from level $\zeta$. 
Hence, the expected number of vertices accepted by the algorithm is at most
\[
\sum_{i=1}^{\zeta-1} \Pr[X_i=1] + 2\,\Pr[X_\zeta=1]
\le
\sum_{i=1}^{\zeta-1} 2^{-(i-1)} + 2 \cdot 2^{-(\zeta-1)}
< 2.
\]
Combining the above, the expected competitive ratio of any randomized online algorithm on $G_\zeta$ is at least $(\zeta+1)/2$.

    Now it remains to show that the graph $G_\zeta$ is a graph with the independent kissing number at most $\zeta$.
    \begin{clm}\label{clm_zeta}
        Graph $G_\zeta$ is a graph with independent kissing number at most $\zeta$.
    \end{clm}
    \begin{proof}
        Fix any vertex $v \in V(G_\zeta)$. By construction, the neighbors of $v$ include at most one vertex from each level, and any two neighbors from the same level are not adjacent. Hence, the neighborhood of $v$ contains no independent set of size greater than $\zeta$, implying that the independent kissing number of $G_\zeta$ is at most $\zeta$, as required.
    \end{proof}
    Hence the lemma follows.
\end{proof}

Applying~\Cref{lem_RLB} to the families of geometric objects for which the bound on the independent kissing number is known due to ~\cite[Thm.~7]{DeKS26}, we obtain the following
corollary that summarizes the implications of our randomized lower bounds in the absence of geometric representations.

\begin{corollary}
When the geometric representation of objects is not available, no randomized online
algorithm can achieve a competitive ratio better than the following
\begin{itemize}
    \item $6.5$ for the online MIS problem on unit ball graphs in
    $\mathbb{R}^3$;
    \item $2^{d-1}+1/2$ for the online MIS problem on geometric intersection graphs of translated copies of a hypercube in $\IR^d$;
    \item $3$ for the online MIS problem on geometric intersection graphs of translated copies of a regular $k$-gon $(k\in\mathbb{Z}^{+}\setminus\{1,2,4\})$
    \item $2^d+1/2$ for the online MIS problem on geometric intersection graphs of congruent hypercubes in $\IR^d$;
    \item $o(\alpha\cdot M^d)$ for the online MIS problem on geometric intersection graphs of $\alpha$-fat objects in $\IR^d$ having widths in the interval $[1,M]$.
\end{itemize}
\end{corollary}

\section{Randomization with Geometric Representation: For Unit Balls in $\IR^3$}\label{sec_disk}

{In this section, we prove~\Cref{thm:rand-alpha}, which shows that randomization, together with access to the geometric representation, allows one to surpass the deterministic lower bound for the online MIS problem on unit ball graphs in $\mathbb{R}^3$.}
Caragiannis et al.~\cite{CaragiannisFKP07} showed that randomization with access to the geometric representation can surpass the deterministic lower bound of $5$ for unit disk graphs.
Their approach is based on a fixed triangular lattice/grid with well-separated points. The periodic structure of the lattice induces a natural partition of the plane into rectangles $R$ of the same size. 
The algorithm applies a uniform random shift to the lattice and processes the input with respect to the shifted lattice. Let $U$ be the union of unit disks centered at the shifted lattice points. An input disk is accepted if its center lies in $U$.
A key \emph{area preserving property} of this construction is that for every region $R$, we have $\mathrm{Area}(R \cap U)=\mathrm{Area}(\text{unit disk})=\pi$.
Accepted disks associated with the same lattice point overlap and form a clique, while disks associated with different lattice points are independent. Running $\AGDS$ on the accepted disks yields an expected competitive ratio of $4.41$.

{  Though a natural  extension of the lattice/grid from two to three dimension is easy to obtain, generalizing their approach from unit disk graphs to unit ball graphs in $\IR^3$ faces  challenges.
\begin{itemize}
  \item[\textbf{Ch1.}] \textbf{Preserving key volumetric properties.}
In higher dimensions, the extension of the key area preserving property is no longer immediate, as the interactions between balls centered at different lattice points become more complicated. In particular, for each suitable hyper-rectangle $R$,  it is not obvious that $\Vol(R \cap U)=\Vol(\text{unit ball in $\IR^3$})$. Establishing this exact identity, formalized in~\Cref{cor-grid-alpha-3d} is crucial for the analysis of the algorithm.

\item[\textbf{Ch2.}] \textbf{Efficiently testing acceptance under random shifts.}
Beyond the volumetric challenge, the randomized approach also raises an algorithmic challenge of  efficiently determining whether the center of an input unit ball in $\IR^3$ lies in $U$. While this condition is implicit in the work of Caragiannis et al.~\cite{CaragiannisFKP07}, making it explicit for higher dimension requires new insight. In~\Cref{lem_verify}, we show that for a given center $c$ of the input unit ball, a closest lattice point $p$ can be computed in constant time. Then the acceptance can be tested by checking whether $\dist(p, c) \le 1$, which can be done in constant time.
\end{itemize}}

Before presenting our algorithm and its analysis, we introduce a lattice and establish its key geometric and volume preserving properties, which play a central role in both its design and analysis.

\subsection{A Lattice and its Properties}\label{sec_lattice}

\smallskip\noindent
\textbf{\underline{Lattice:}} 
Let us define $d$ linearly independent vectors ${\bf v}_1,{\bf v}_2,\cdots,{\bf v}_d$ as follows.
Let ${\bf v}_1=(4+\delta){\bf e}_1$, and for all $i\in[d]\setminus\{1\}$, let ${\bf v}_i =-(2+\delta/2){\bf e}_1+2\sqrt{3}{\bf e}_i$, where, $\delta>0$ is an arbitrarily small fixed constant and ${\bf e}_1,{\bf e}_2,{\bf e}_3$ are the standard unit vectors in $\IR^d$. 
 The lattice generated by these vectors is
$$\Lambda_d=\{\alpha_1 {\bf v}_1+\alpha_2 {\bf v}_2+\cdots+\alpha_d {\bf v}_d\ |\ (\alpha_1,\alpha_2,\cdots,\alpha_d)\in \mathbb{Z}^d\}.$$ For notational simplicity, we denote $\Lambda_3$ by $\Lambda$.  A visualization of the lattice $\Lambda$ is shown in~\Cref{fig:lattice}. 

\begin{figure}[htbp]
    \centering
\includegraphics[page=2,scale=0.9]{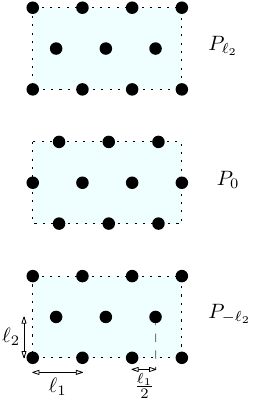}
       \caption{Points of $\Lambda$ are drawn. Here $\ell_1=4+\delta$, and $\ell_2=2\sqrt{3}$. The projections of planes $P_{\ell_2}$ and $P_{0}$, $P_{-\ell_2}$ over a rectangular region are depicted, where $P_k$ denotes the plane parallel to $xy$ plane with $z$-coordinate value $k$. There is no lattice point in between any of these two planes. }
       \label{fig:lattice}
\end{figure}

\smallskip\noindent
\textbf{\underline{Properties of the lattice:}}
\begin{restatable}{lemma}{lemDistLambda}\label{obs-lattice-Lambda}
The distance between any two distinct points in the lattice $\Lambda$ is strictly greater than~$4$.
\end{restatable}
\begin{proof}
Let $p_1$ and $p_2$ be two distinct points in $\Lambda$. Since $p_1 \neq p_2$, they differ in at least one coordinate. 
If the difference in any coordinate is at least $4+\delta$, then  $\dist(p_1,p_2) > 4$ and we are done.
Otherwise, by the construction of the lattice $\Lambda$, the smallest possible nonzero differences between corresponding coordinates are as follows: in the first coordinate, the difference is at least $2+\delta/2$, and in each of the remaining coordinates it is at least $2\sqrt{3}$.
Suppose first that $p_1$ and $p_2$ differ in first coordinate  by $2+\delta/2$. Then, by the definition of $\Lambda$, they must also differ in at least one other coordinate by at least $2\sqrt{3}$. Similarly, if they differ in some coordinate $i \in \{2,3\}$ by $2\sqrt{3}$, then they must differ in the first coordinate by at least $2+\delta/2$. In either case, the squared Euclidean distance between $p_1$ and $p_2$ is at least
$(2+\delta/2)^2 + (2\sqrt{3})^2 > 16 $
and hence $\dist(p_1,p_2) > 4$. This completes the proof.
\end{proof}

\subparagraph*{\underline{Key volumetric property:}} The following theorem establishes the volume-preserving property that plays a crucial role in the proof of~\Cref{thm:rand-alpha}.
\begin{theorem}\label{cor-grid-alpha-3d}
Let $\sigma$ be a unit ball in $\mathbb{R}^3$, and let $U=\cup_{p\in\Lambda} \sigma(p)$ denote the union of translated copies of $\sigma$ centered at all the lattice points of $\Lambda$.
Fix $l=(l(x_1),l(x_2),l(x_3))\in\mathbb{R}^3$ and define the axis-aligned box $R=\Big[l(x_1),l(x_1)+4+\delta\Big]\times\prod_{i=2}^3\Big[l(x_i),l(x_i)+2\sqrt{3}\Big]$, such that the side length along the $x_1$-axis is $4+\delta$ and the side lengths along the $x_2$- and $x_3$-axes are $2\sqrt{3}$. 
Then, we have $\Vol(R\cap U)=\Vol(\sigma)$.
\end{theorem}

We first present a few definitions, observations, and claims related to the structure of the lattice $\Lambda$ that we need in proving~\Cref{cor-grid-alpha-3d}.

The lattice $\Lambda$ is periodic along the coordinate axes: its period along the $x_1$-axis is $4+\delta$, and along each of the $x_2$- and $x_3$-axes it is $4\sqrt{3}$.
In particular, if two lattice points $p_1,p_2\in \Lambda$ satisfy $|p_1(x_i)-p_2(x_i)|=c\neq 0$ for some $i\in\{1,2,3\}$ and $|p_1(x_j)-p_2(x_j)|=0$ for all $j\neq i$, then $c$ must be an integral multiple of the period of $\Lambda$ along the $i$th axis. 

We say that an axis-aligned box $H=\prod_{i=1}^3[l(x_i),u(x_i)]$ is \emph{periodic in the $i$th axis} with respect to $\Lambda$ if the length of $[l(x_i),u(x_i)]$ is an integral multiple of the period of $\Lambda$ along that axis.

\begin{observation}\label{obs:invar-3d}
If $H$ is periodic in the $i$th axis with respect to $\Lambda$, then the volume of $(H\cap U)$ is invariant under translations of $H$ along the $i$th axis. That is,
$ \Vol(H\cap U)=\Vol((H+v{\bf e}_i)\cap U)$ for every $v\in\mathbb{R}$ and every $i\in[3]$.
\end{observation}

%

Let $$R_0=\Big[l(x_1),l(x_1)+(2+\delta/2)\Big]\times\prod_{i=2}^3\Big[l(x_i),l(x_i)+2\sqrt{3}\Big]$$ be an axis-aligned box in $\IR^3$, and let $R_1=R_0+(2+\delta/2){\bf e}_1$ be the box obtained by translating $R_0$ by $2+\delta/2$ along the $x_1$-axis. Clearly, $R=R_0\cup R_1$.  

\begin{clm}\label{clm_sure_corner}
If a corner point of $R$ coincides with a lattice point in $\Lambda$, then there exist corners of both $R_0$ and $R_1$ that coincide with points in $\Lambda$; see~\Cref{fig_clm_sure_corner}.
\end{clm}
\begin{proof}
   Let $x$ be a corner point of $R$ such that $x\in\Lambda$.
   Since $R=R_0\cup R_1$, the point $x$ must be a corner of either $R_0$ or $R_1$. Without loss of generality, assume that $x$ is a corner of $R_0$. By the construction of $R_0$ and $R_1$, 
 exactly one of the points
$$
x + \big((2+\delta/2){\bf e}_1 + 2\sqrt{3}{\bf e}_2\big)
\quad\text{or}\quad
x + \big((2+\delta/2){\bf e}_1 - 2\sqrt{3}{\bf e}_2\big)
$$
is a corner point of $R_1$. Moreover, both of these points belong to $\Lambda$ whenever
$x\in\Lambda$. Hence, $R_1$ has a corner that coincides with a lattice point, as required. \myqed
\end{proof}

\begin{figure}[htbp]
    \centering
    \includegraphics[page=2,width=0.25\linewidth]{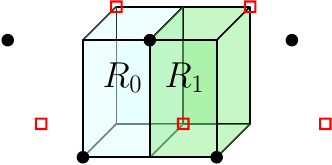}
    \caption{Illustration of~\Cref{clm_sure_corner}. The box $R$ is decomposed into $R_0$ and its translate $R_1$ along the $x_1$-axis. If a corner of $R$ coincides with a lattice point in $\Lambda$ (where $\Lambda$ denotes the set of lattice points), then the construction ensures that both $R_0$ and $R_1$ have corners that also coincide with lattice points. The black points denote lattice points in the plane at $z=0$, while red points denote the lattice points in the plane at $z=2\sqrt{3}$.}
    \label{fig_clm_sure_corner}
\end{figure}

 \begin{observation}\label{clm:corner}
 If a corner point of $R_0$ coincides with a lattice point in $\Lambda$, then exactly four corners of $R_0$ coincide with points of~$\Lambda$; see~\Cref{fig_clm_sure_corner}.
\end{observation}

 Before stating the next claim, we define a translation between the two boxes.
Let $f: R_0\rightarrow R_1$ be a mapping defined by $f(x)=x+(2+\delta/2){\bf e}_1$.

\begin{restatable}{clm}{clmNotLambda}\label{clm:not_lambda}
If a point $x\in R_0$ belongs to the lattice $\Lambda$, then its image $f(x)$ does not belong to $\Lambda$. Conversely, if a point $y\in R_1$ belongs to $\Lambda$, then its pre-image $f^{-1}(y)$ does not belong to $\Lambda$.
\end{restatable}
\begin{proof}
  Let $x\in R_0\cap\Lambda$. Since $\Lambda$ is generated by ${\bf v}_1,{\bf v}_2,{\bf v}_3$, we can write
 $x=z_1{\bf v}_1+z_2{\bf v}_2+z_3{\bf v}_3$, for some $(z_1,z_2,z_3)\in\mathbb{Z}^3$.
    By the definition of $f$, we have
    $       f(x)=x+(2+\delta/2){\bf e}_1
     =\left(z_1+1/2\right){\bf v}_1+z_2{\bf v}_2+z_3{\bf v}_3$.
    Since $z_1\in\mathbb{Z}$, the scalar $\left(z_1+1/2\right)\notin\mathbb{Z}$. Hence, $f(x)\notin\Lambda$.

Now, we will show that if a point $y=f(x')\in R_1$ belongs to the lattice $\Lambda$, then the pre-image $f^{-1}(y)=x'$ does not belong to $\Lambda$. If $y\in R_1\cap\Lambda$, and $y=f(x')=x'+(2+\delta/2){\bf e}_1=z_1{\bf v}_1+z_2{\bf v}_2+z_3{\bf v}_3$, where $(z_1,z_2,z_3)\in\mathbb{Z}^3$.
    Note that when the point $y=f(x')=x'+(2+\delta/2){\bf e}_1\in{R}_1$ then the pre-image $f^{-1}(y)=f^{-1}(f(x'))=x'\in{R}_0$. We have
        $x'=y-{(2+\delta/2){\bf e}_1}
        =\left(z_1-1/2\right){\bf v}_1+z_2{\bf v}_2+z_3{\bf v}_3$.
    Since $z_1\in\mathbb{Z}$, the scalar $\left(z_1-1/2\right)\notin\mathbb{Z}$. Thus, we have $x'\notin\Lambda$. Hence, the claim follows.
\myqed
\end{proof}

 For each $i\in\{2,3\}$, let $h_i: R_1\rightarrow R_i$ be a mapping defined by
 $h_i(x)=x-(2+\delta/2){\bf e}_1+2\sqrt{3}{\bf e}_i$.

\begin{clm}\label{clm:lambda}
 For each $i\in\{2,3\}$, a point $x\in R_1$ belongs to the lattice $\Lambda$ if and only if its image $h_i(x)$ belongs to $\Lambda$.
\end{clm}
\begin{proof}
      We first prove the forward direction.
    Let $x\in R_1\cap\Lambda$. Since $x\in\Lambda$, we have $x=z_1{\bf v}_1+z_2{\bf v}_2+z_3{\bf v}_3$, where $(z_1,z_2,z_3)\in\mathbb{Z}^3$.
    By the definition of $h_i$, we have
    \begin{align*}
        h_i(x)=&x-(2+\delta/2){\bf e}_1+2\sqrt{3}{\bf e}_i=z_1{\bf v}_1+z_2{\bf v}_2+z_3{\bf v}_3+\left(-(2+\delta/2)\right){\bf e}_1+2\sqrt{3}{\bf e}_i\\
        =&z_1{\bf v}_1+z_2{\bf v}_2+z_3{\bf v}_3+{\bf v}_i\tag{Since ${\bf v}_i=\left(-(2+\delta/2)\right){\bf e}_1+2\sqrt{3}{\bf e}_i$}\\
        =&z_1{\bf v}_1+(z_2+1){\bf v}_2+z_3{\bf v}_3 \text{ or }z_1{\bf v}_1+z_2{\bf v}_2+(z_3+1) {\bf v}_3.
    \end{align*}
    Since $z_i\in\mathbb{Z}$, the scalar $z_i+1$ is also an integer point. Hence, $h_i(x)\in\Lambda$. In a similar way, one can prove the converse part.\myqed
\end{proof}


\begin{clm}\label{clm_R_same_R'}
    For each $i\in\{2,3\}$, the volume of $R\cap U$ is equal to the volume of $R_i'(=R_0\cup R_i)\cap U$.
\end{clm}

\begin{proof}
    By~\Cref{clm:lambda}, a point $p\in{R_1}$ belongs to the lattice $\Lambda$ if and only if the point $p+\left(-(2+\delta/2){\bf e}_1+2\sqrt{3}{\bf e}_i\right)$ is a point in $\Lambda$ and lies in $R_i$. Consequently, the set of lattice points inducing balls intersecting $R_1$ is in one-to-one correspondence with the set of lattice points inducing balls intersecting $R_i$, and hence we have $\Vol(R_1\cap U)=\Vol( R_i\cap U)$. It follows that, for each $i\in\{2,3\}$,  we have 
$\Vol(R\cap U)=\Vol( (R_0\cup R_1)\cap U)=\Vol((R_0\cup R_i)\cap U)=\Vol(R_i'\cap U)$, as claimed.
\end{proof}

Using the claims and observations established above, we now prove~\Cref{cor-grid-alpha-3d}.\\


\begin{proof}[Proof of~\Cref{cor-grid-alpha-3d}.]
First, we show that if a corner point of $R$ coincides with a lattice point, then $\Vol(R\cap U)=\Vol(\sigma)$. We then prove that $\Vol(R\cap U)$ is invariant under any translation.  
Let $x$ be a corner of $R$ such that $x\in\Lambda$. Due to~\Cref{clm_sure_corner}, both $R_0$ and $R_1$ have corners that also coincide with $\Lambda$. Let $P_0=R_0\cap U$ and $P_1=R_1\cap U$.
Observe that three mutually orthogonal planes passing through a point $c$ partition $\IR^3$ into eight octants. Consequently, the unit ball $\sigma$ centered at $c$ can be decomposed into at most eight distinct parts. By~\Cref{clm:corner}, each of $P_0$ and $P_1$ contains four such disjoint parts of $\sigma$, respectively. Moreover, by~\Cref{clm:not_lambda}, the parts contained in $P_0$ and $P_1$ are distinct. Hence, combining the parts from $P_0$ and $P_1$ recovers the entire ball $\sigma$. Therefore, when a corner of $R$ lies in $\Lambda$, we have $\Vol(R\cap U)=\Vol(\sigma)$.

We now show that $\Vol( R\cap U)$ is invariant under translations. 
Let $\bf v\in\IR^3$ be an arbitrary translation vector, which we write as ${\bf v}=({\bf v}(x_1),{\bf v}(x_2),{\bf v}(x_3))={\bf v}(x_1){\bf e}_1+{\bf v}(x_2){\bf e}_2+{\bf v}(x_3){\bf e}_3$. 
Since any translation can be decomposed into successive translations along the coordinate axes, it suffices to prove invariance under translations of the form ${\bf v}(x_i){\bf e}_i $, for each $i\in[3]$. 
 We begin with translations along the $x_1$-axis. The side length of $R$ along the $x_1$-axis equals the period of the lattice $\Lambda$ in that direction. Therefore, by~\Cref{obs:invar-3d}, $\Vol( R\cap  U) =\Vol(( {R +{\bf v}(x_1){\bf e}_1})\cap U)$. 
Next, consider translations along the $x_i$-axis for $i\in\{2,3\}$. 
 Recall that for each $i\in\{2,3\}$, we have $R_i'=R_0\cup R_i$, where $R_i=R_1-(2+\delta/2){\bf e}_1+2\sqrt{3}{\bf e}_i$ is a translated copy of $R_1$. 
We then observe that 
$$ 
\Vol(R\cap U) =\Vol( R_i'\cap U) =\Vol( (R_i'+{\bf v}(x_i){\bf e}_i)\cap U) =  \Vol( ( R+{\bf v}(x_i){\bf e}_i)\cap U).$$
 Here, the first and third equalities follow from~\Cref{clm_R_same_R'}, while the second equality follows from~\Cref{obs:invar-3d}. 
Hence, we  conclude that $\Vol(R\cap U)$ is invariant under arbitrary
translations. This completes the proof. \myqed
\end{proof}


\begin{rmk}
    The lattice properties established above for $d=3$ admit analogues in $d \ge 4$.
\end{rmk}

\subsection{Randomized Algorithm and its Analysis}
Throughout this section, all balls are assumed to be unit balls. A unit ball centered at a point $c\in\IR^3$ is denoted by $\sigma(c)$.
We consider a randomized algorithm obtained by generalizing the framework of Caragiannis et al.~\cite{CaragiannisFKP07}. In the framework, randomization can be implemented either by randomly shifting the underlying lattice or, equivalently, by applying a fixed random shift to the centers of the input balls. Both viewpoints yield the same relative configuration between the lattice and the input objects. For simplicity of presentation and analysis, we adopt the latter viewpoint.

\smallskip\noindent
\textbf{\underline{Description of the Algorithm $\AF$:}}
At the beginning, the algorithm $\AF$ selects three scalars $\beta_1\in[0,4+\delta)$ and $\beta_2,\beta_3 \in[0,2\sqrt{3})$ independently and uniformly at random, and forms the vector ${\bf b}=(\beta_1,\beta_2,\beta_3)$, where $\delta>0$ is an arbitrarily small fixed constant. Upon the arrival of a ball $\sigma(c)$, the algorithm proceeds as follows. If there exists a lattice point $p\in\Lambda$ such that $\dist(p,c+{\bf b})\leq 1$, then it executes the algorithm $\AGDS$ on $\sigma(c)$; Otherwise, it ignores $\sigma(c)$.

\subparagraph*{\underline{Efficiently testing acceptance under random shifts.}}

Before presenting the analysis of the algorithm, it is interesting to note that one can verify, in constant time, whether there exists a lattice point $p\in\Lambda$ such that the distance between the shifted center $c+{\bf v}$ of an input ball $\sigma(c)$ and $p$ is at most 1.

\begin{restatable}{lemma}{lemVerify}\label{lem_verify}
Let $c'=c+\mathbf{b}\in\mathbb{R}^3$, where $c$ is the center of a unit ball $\sigma(c)$ and $\mathbf{b}=(\beta_1,\beta_2,\beta_3)$ with $\beta_1\in[0,4+\delta)$ and $\beta_2,\beta_3\in[0,2\sqrt{3})$. 
Then there exists a lattice point $p\in\Lambda$ that is closest to $c'$, and such a point $p$ can be computed in $O(1)$ time. Consequently, deciding whether $\dist(p,c')\le 1$ can be done in $O(1)$ time.
\end{restatable}
\begin{proof}

Recall that $$\Lambda=\{\alpha_1 {\bf v}_1+\alpha_2 {\bf v}_2+\alpha_3 {\bf v}_3\ |\ (\alpha_1,\alpha_2,\alpha_3)\in \mathbb{Z}^3\},$$ where ${\bf v}_1=(4+\delta){\bf e}_1$, ${\bf v}_2 =-\left(2+\delta/2\right){\bf e}_1+2\sqrt{3}{\bf e}_2$,  and ${\bf v}_3 =-\left(2+\delta/2\right){\bf e}_1+2\sqrt{3}{\bf e}_3$.
Thus, any lattice point $y\in\Lambda$ can be written explicitly as $$y=\left((4+\delta)a_1-\left(2+\delta/2\right)a_2-\left(2+\delta/2\right)a_3\right){\bf e}_1+2\sqrt{3}a_2{\bf e}_2+2\sqrt{3}a_3{\bf e}_3,$$ for integers $a_1,a_2,a_3$.
Given a unit ball $\sigma(c)$, the algorithm $\AF$ forms the shifted center $c'=c+\mathbf{b}$, where $\mathbf{b}=(\beta_1,\beta_2,\beta_3)$ with $\beta_1\in[0,4+\delta)$ and $\beta_2,\beta_3\in[0,2\sqrt{3})$.
For each $i\in[3]$, write $$c'(x_i)=z_i\sqrt{3}+y_i,$$ where $z_i\in\mathbb{Z}$ and $y_i\in\left[0,\sqrt{3}\right)$.

\medskip
\noindent
\textit{\underline{Determining the second and third coordinates of the closest lattice point from $c'$.}}\\
For each $i\in\{2,3\}$, the $i$-th coordinate of any lattice point in $\Lambda$ is an integer multiple of $2\sqrt{3}$. We therefore define
$$
p(x_i)=
\begin{cases}
    z_i\sqrt{3},& \text{if $z_i$ is even}\\
   (z_i+1)\sqrt{3}, & \text{if $z_i$ is odd}.
\end{cases}$$
By construction, $p(x_i)$ is the closest admissible lattice coordinate to $c'(x_i)$, for $i\in\{2,3\}$.

Observe that once $p(x_i)$ is fixed for $i=2,3$, the parity of the corresponding lattice coefficient $a_i$ is uniquely determined, since $p(x_i)=2\sqrt{3}\,a_i$ and the algorithm restricts $a_i$ to be even or odd according to the parity of $z_i$. Hence the value $k \equiv a_2 + a_3 \pmod{2}$ is known before choosing $p(x_1)$.

\medskip
\noindent
\textit{\underline{Determining the first coordinate of the closest lattice point from $c'$.}}\\
Observe that the first coordinate of a lattice point $p=a_1{\bf v}_1+a_2{\bf v}_2+a_3{\bf v}_3$ equals $(2+\delta/2)(2a_1-a_2-a_3)$. Since $2a_1$ is always even, the parity of the integer
multiplier of $(2+\delta/2)$ is completely determined by $a_2$ and $a_3$. 
Let $$k := |\{i\in\{2,3\} : a_i \text{ is odd}\}|.$$
Write $c'(x_1)=\left(2+\delta/2\right)z_1+y_1\in\IR$, where $z_1\in\mathbb{Z}$ and $y_1\in\left[0,\left(2+\delta/2\right)\right)$.
We define 
$$
p(x_1)=
\begin{cases}
(2+\delta/2)z_1, & \text{if $z_1 \equiv k \pmod{2}$},\\
(2+\delta/2)(z_1+1), & \text{otherwise}.
\end{cases}
$$
By construction, the resulting point $p$ belongs to $\Lambda$, and it can be computed in $O(1)$ time.
We will prove that $p$ is the closest lattice point to $c'$.
\begin{clm}\label{clm_closest}
    The point $p$ is the closest point of $\Lambda$ to $c'$.
\end{clm}
\begin{proof}
    Let $p'\neq p$ be any other point in $\Lambda$. Then, we have
   $$
        p'(x_1)=(2+\delta/2)(2a_1-a_2-a_3),$$
        $$
p'(x_2)=2\sqrt{3}a_2,$$
$$
p'(x_3)=2\sqrt{3}a_3,
    $$
for integers $a_1,a_2,a_3$.

For each $i\in\{2,3\}$, by construction, $p(x_i)$ is the closest admissible multiple of $2\sqrt{3}$ to $c'(x_i)$, and hence, we have
$$
|p'(x_i)-c'(x_i)| \ge |p(x_i)-c'(x_i)|.
$$

For the first coordinate, the admissible values of $p'(x_1)$ are exactly those integer multiples of $(2+\delta/2)$ whose parity is determined by $k$.
By definition, $p(x_1)$ is the closest such admissible value to $c'(x_1)$, and therefore, we have
$$
|p'(x_1)-c'(x_1)| \ge |p(x_1)-c'(x_1)|.
$$

Combining the inequalities for all three coordinates and using the monotonicity of the Euclidean norm with respect to coordinate-wise absolute differences, we obtain
$$
\dist(p',c') \ge \dist(p,c').
$$
Thus, $p$ is the closest point of $\Lambda$ to $c'$.
Hence, the claim follows.
\myqed
\end{proof} 
Hence, the closest lattice point to $c'$ can be computed in $O(1)$ time, and whether $\dist(p,c')\leq 1$ can also be performed in $O(1)$ time, as required.      
\end{proof}

\smallskip\noindent
\textbf{\underline{Analysis of the Algorithm $\AF$:}}
We now analyze the performance of the algorithm $\AF$, thereby completing the proof of~\Cref{thm:rand-alpha}.
 Let $\cal I$ denote the sequence of unit balls presented to $\AF$, and let ${\cal I}'\subseteq {\cal I}$ be the sub-sequence consisting of those balls on which $\AF$ executes $\AGDS$. Let $\OO$ be an MIS for $\cal I$, and $\A$ be an IS returned by $\AF$. 
\begin{lemma}\label{clm:alpha_prob}
    The algorithm $\AGDS$ is executed on a unit ball $\sigma\in{\cal I}$ with  probability $\left(\frac{(4/3)\pi}{48+12\delta}\right)$, where $\delta>0$ is an arbitrary small constant.
\end{lemma}
\begin{proof}
 Recall that $\AF$ executes $\AGDS$ on a unit ball $\sigma(c)$, if there exists a lattice point $p\in\Lambda$ such that $\dist(p,c')\leq 1$, where $c'=c+{\bf b}$. 
Let $ U=\cup_{p\in\Lambda} {\sigma}(p)$ denote the union of unit balls centered at the lattice points.
Define a 3-dimensional axis-aligned box 
$$R=\left[c(x_1),c(x_1)+4+\delta\right]\times\prod_{j=2}^3\left[c(x_j),c(x_j)+2\sqrt{3}\right].$$ 
By~\Cref{cor-grid-alpha-3d}, we have $\Vol(R\cap  U)=\Vol(\sigma)$.
 Since ${\bf b}$ is chosen uniformly at random, the shifted center $c'=c+{\bf b}$ is
uniformly distributed in $R$. The algorithm $\AF$ executes $\AGDS$ on $\sigma$ exactly
when $c'\in U$. Therefore, the probability that $\AGDS$ is executed on $\sigma\in{\cal I}$ is equal to the ratio $\Vol(\sigma)/ \Vol(R)$, which is~$\left(\frac{(4/3)\pi}{48+12\delta}\right)$. 
 Hence, the lemma follows.
   \end{proof}

Let $\OO\cap{\cal I}'$ be the set consisting of unit balls of $\OO$ on which $\AGDS$ is executed, and let $\OO'$ be an MIS for ${\cal I}'$. Since $\OO\cap{\cal I}'$ is an independent set in ${\cal I}'$, we have $|\OO'|\geq |\OO\cap{\cal I}'|$. By the linearity of expectation, we have 
$\mathbb{E}[|\OO'|]\geq\mathbb{E}[|\OO\cap{\cal I}'|]$.
For $\sigma\in\OO$, define a random variable $X_{\sigma}$ that equals $1$ if
$\AGDS$ is executed on $\sigma$ and $0$ otherwise. Then, we have
$$\mathbb{E}[|\OO\cap{\cal I}'|]=\sum_{\sigma\in\OO} Pr[ \AGDS \text{ is executed on } \sigma]\cdot X_{\sigma}
   \geq\left(\frac{(4/3)\pi}{48+12\delta}\right)|\OO|,$$
where the last inequality follows from~\Cref{clm:alpha_prob}.
Since $\mathbb{E}[|\OO\cap{\cal I}'|]\geq\left(\frac{(4/3)\pi}{48+12\delta}\right)|\OO|$, we have
    $\mathbb{E}[|\OO'|]\geq\left(\frac{(4/3)\pi}{48+12\delta}\right)|\OO|$.

\begin{lemma}\label{conn-comp-alpha}
    Each connected component of the unit ball graph of ${\cal I}'$ is a clique.
\end{lemma}

\begin{proof}
By construction, $\AF$ executes $\AGDS$ on a ball $\sigma(c)$ if and only if the shifted
center $c'=c+{\bf v}$ lies in $\sigma(p)$ for some $p\in\Lambda$. Thus, ${\cal I}'$ consists precisely of those balls whose shifted centers lie inside some $\sigma(p)$, where$p\in\Lambda$.
We show that two balls $\sigma(c_1)$ and $\sigma(c_2)$ in ${\cal I}'$ intersect if and only if their shifted centers $c_1'=c_1+{\bf b}$ and $c_2'=c_2+{\bf b}$ lie in the same $\sigma(p)$ for some $p\in\Lambda$.

For the forward direction, we prove the contrapositive statement:  ``If $c_1'$ and $c_2'$ lies in $\sigma(p_1)$ and $\sigma(p_2)$, where $p_1,p_2(\neq p_1)\in\Lambda$, then objects $\sigma(c_1)$ and $\sigma(c_2)$ do not intersect''. 
Due to~\Cref{obs-lattice-Lambda}, we have $\dist(p_1,p_2)>4$. By the triangle inequality, we have $\dist(p_1,c_1')+\dist(c_1',c_2')+\dist(c_2',p_2)> 4$. Note that $\dist(p_1,c_1')\leq 1$ and $\dist(c_2',p_2)=\dist(p_2,c_2')\leq 1$. Thus, we have $\dist(c_1',c_2')>2$, implying $\sigma(c_1')\cap\sigma(c_2')=\emptyset$. Since both centers are shifted by the same translation vector, we also have $\sigma(c_1)\cap\sigma(c_2)=\emptyset$. This completes the forward part.

Conversely, for some $p\in\Lambda$ if $c_1'$ and $c_2'$ lie in a same $\sigma(p)$, then $\dist(p,c_1')=\dist(c_1',p)\leq 1$  and $\dist(p,c_2')=\dist(c_2',p)\leq 1$, implying $\sigma(c_1')\cap\sigma(c_2')\neq \emptyset$. Since ${\bf b}$ is added to both $c_1$ and $c_2$ to obtain $c_1'$ and $c_2'$, we have $\sigma_1(c_1)\cap\sigma_2(c_2)\neq\emptyset$.
Hence, the claim follows. \myqed
\end{proof}

Since each connected component of the unit ball graph in ${\cal I}'$ is a clique,  any independent set in $\I'$ contains at most one unit ball from each connected component. The algorithm $\AF$, however, selects exactly one unit ball from each connected component. Thus, $|\A|\geq|\OO'|$, and by linearity of expectation, we have $\mathbb{E}[|\A|]\geq\mathbb{E}[|\OO'|]$.
Combining this with the bound on $\mathbb{E}[|\mathcal{O}'|]$ yields 
$\mathbb{E}[|\A|]\geq\left(\frac{(4/3)\pi}{48+12\delta}\right)|\OO|$. This completes the proof of~\Cref{thm:rand-alpha}.


\subsection*{Extension of algorithm $\AF$ for unit balls in $\IR^d$, for $d\geq 4$}\label{sect_extension}
The lattice-based randomized framework underlying $\AF$ can also be formulated for unit
balls in $\IR^d$ with $d\ge 4$. The algorithmic construction and the geometric arguments
used in the three-dimensional case can be generalized in an analogous manner in higher
dimensions, leading to a corresponding algorithm in $\IR^d$. 

However, the performance of the generalized algorithm deteriorates quickly as the
dimension increases. The main reason is geometric: the volume of the unit ball in $\IR^d$
decreases rapidly with $d$, while the volume of the corresponding fundamental box grows by
a factor of $2\sqrt{3}$ in each additional dimension. As a consequence, the probability
that a randomly shifted ball center lies within distance $1$ of a lattice point decreases
sharply with $d$, resulting in larger competitive ratios.

The structural properties proven for the three-dimensional case (see~\Cref{sec_lattice}) extend analogously to higher dimensions. Using these properties, the competitive ratio of $\AF$ against an oblivious adversary in $\IR^d$ is at most $((4+\delta)(2\sqrt{3})^{d-1})/\Vol(\text{unit ball in }\IR^d)$.
For $d=4$, this bound evaluates to $\frac{192\sqrt{3}}{\pi^2}+\varepsilon \approx 33.8,
\quad\text{where}\quad
\varepsilon=\frac{48\sqrt{3}\delta}{\pi^2}$.
Since the deterministic algorithm $\AGDS$ achieves a competitive ratio of at most $24$
for unit balls in $\IR^4$~\cite{Musin03}, the generalized algorithm $\AF$ does not surpass
the known deterministic lower bound in this dimension. This gap only widens for larger
values of $d$, highlighting a limitation of the current approach.

\section{Randomization with Geometric Representation: For Similarly Sized Fat Objects in $\IR^d$}\label{sec_fat}
In this section, we prove~\Cref{thm:ACY} which is for the online MIS problem on geometric
intersection graphs of $\alpha$-fat objects in $\IR^d$ whose widths lie in the range $[1, M]$, where  $M>2$. In addition to the standard online model, the algorithm assumes that the width of each object is revealed upon arrival and that the value of $M$ is known in advance.

\smallskip\noindent
\textbf{\underline{Description of the Algorithm $\ACY$:}}
The algorithm partitions the input objects into $\lfloor \log M \rfloor + 1$ disjoint
classes according to their widths. For each
$j\in\{0,1,\ldots,\lfloor \log M \rfloor\}$, let $S_j$ denote the set of objects whose
widths lie in the interval $[2^j,2^{j+1})$. 
When the first object arrives, the algorithm selects an index $j$ uniformly at random from
$\{0,1,\ldots,\lfloor \log M\rfloor\}$ and executes $\AGDS$ only on objects belonging to
the class $S_j$, rejecting all other objects.

\smallskip\noindent
\textbf{\underline{Analysis of the Algorithm $\ACY$:}}
We now analyze the performance of the algorithm $\ACY$. The following proof establishes the competitive ratio claimed in~\Cref{thm:ACY}.\\

\begin{proof}[Proof of~\Cref{thm:ACY}.]
Let $\OO$ be an MIS for an input sequence $\cal I$. For $j\in\{0,1,2,\ldots$, $\lfloor \log M\rfloor\}$, let ${S}_{j}$ be the collection of all $\alpha$-fat objects in $\cal I$ having widths in  the range $\left[2^j, 2^{j+1}\right)$. Let $\OO_j$ be an MIS for the objects belonging to the set $S_j$. Since $\OO\cap S_{j}$ is an independent set for the objects in ${S}_j$, we have $|\OO_j|\geq |\OO\cap {S}_j|$. 
 
Let $\A$ be the set of objects accepted by the algorithm $\ACY$, and let $\A_j$ be the set of objects accepted by the algorithm $\ACY$, assuming that the set $S_j$ is selected and the algorithm $\AGDS$ is executed on the objects of set $S_j$.
 Due to~\Cref{thm:ind}, $\ACY$ accepts at least $|\A_j|\geq \frac{1}{\zeta_{j}}|\OO_j|$ objects of $S_j$, where $\zeta_{j}$ is the independent kissing number of $\alpha$-fat objects having widths in  the range $[2^j,2^{j+1})$. 
Observe that, for each $j$, scaling down the objects in $S_j$ by a factor of $\frac{1}{2^j}$
to obtain objects in $S_0$ does not change the geometric intersection graph.
  Thus, for each $j$, the value of $\zeta_{j}$ is the same as $\zeta'$. 
 Since $|\OO_j|\geq |\OO\cap S_j|$, we have $|\A_j|\geq  \frac{1}{\zeta'}|\OO\cap S_j|$. 
Note that the set $S_j$ is selected with  probability $\frac{1}{\lfloor\log M\rfloor+1}$, where $j\in\{0,1,2,\ldots,\lfloor \log M\rfloor\}$.  The expected size of $\A$ is 
\begin{align*}
 \mathbb{E}[|\A|]=&\sum_{j=0}^{\lfloor \log M\rfloor} (Pr[S_j\text{ is selected}])\times |\A_j|\geq 
 \{Pr[S_j\text{ is selected}]\} \times \sum_{j=0}^{\lfloor \log M\rfloor}  \frac{1}{\zeta'}|\OO\cap S_j|\\
 \geq&\{Pr[S_j\text{ is selected}]\}\times \frac{1}{\zeta'} |\OO| = \frac{1}{\zeta' (\lfloor\log_2M\rfloor+1)}|\OO|.   
\end{align*}
Hence, the theorem follows. \myqed \end{proof}


{Note that our algorithms also apply to the special case $M=1$, corresponding to unit-size objects. However, in this regime, our performance guarantees do not surpass the known randomized lower bounds.}

\section{Randomization with Geometric Representation: For Similarly Sized Hyper-rectangles in $\IR^d$}\label{sec_rect}
In this section, we prove~\Cref{thm:RCY} which is for the geometric intersection graphs of axis-aligned hyper-rectangles in $\IR^d$ whose side lengths lie in the range $[1,M]$, where $M>2$.

\smallskip\noindent
\textbf{\underline{Description of the Algorithm $\RCY$:}} 
Hyper-rectangles in $\IR^d$ having side lengths in $[1,M]$ can be classified into $(\lfloor \log M\rfloor+1)^d$ disjoint classes according to their side lengths as follows.
For $i_1,i_2,\ldots,i_d\in[\lfloor \log M\rfloor]\cup\{0\}$, let $S^{(i_1,\ldots,i_d)}$ be the class containing hyper-rectangles $\sigma=\prod_{j=1}^d[l(x_j),u(x_j)]$ such that for each coordinate $j\in[d]$, we have $|u(x_j)-l(x_j)|\in[2^{i_j},2^{i_j+1})$. Upon the arrival of the first hyper-rectangle, the algorithm selects indices $i_1,i_2,\ldots,i_d$ uniformly at random from the set $\{0,1,2,\ldots,\lfloor \log M\rfloor\}$ and executes $\AGDS$ only on hyper-rectangles belonging to the selected class $S^{(i_1,\ldots,i_d)}$, rejecting all other hyper-rectangles.

\smallskip\noindent
\textbf{\underline{Analysis of the Algorithm $\RCY$:}}
A key ingredient distinguishing the hyper-rectangle setting from the general $\alpha$-fat case is the ability to bound the independent kissing number of each class by a constant that depends only on the dimension $d$ and is independent of $M$. The following lemma establishes this property.
\begin{lemma}\label{IKN_rectangle}
    The independent kissing number $\zeta^{(i_1,\ldots,i_d)}$ for the geometric intersection graph induced by axis-aligned hyper-rectangles in the class $S^{(i_1,\ldots,i_d)}$ is at most $4^d$, where $i_1,\ldots,i_d \in \{0,1,\ldots,\lfloor \log M\rfloor\}$.
\end{lemma}
\begin{proof}
    To establish an upper bound on the value of $\zeta^{(i_1,\ldots,i_d)}$, we consider the problem of
    placing a maximum number of pairwise non-intersecting smallest-volume hyper-rectangles from the class $S^{(i_1,\ldots,i_d)}$ such that each of them
    intersects a largest-volume hyper-rectangle $L$ from the same class.
  Let $L=\prod_{j=1}^d[l(x_j),u(x_j)]$ 
     and $S=\prod_{j=1}^d[l'(x_j),u'(x_j)]$ such that for each $j\in[d]$, we have $|u(x_j)-l(x_j)|<2^{i_j+1}$ and $|u'(x_j)-l'(x_j)|=2^{i_j}$. Thus, $L$ and $S$ are, respectively, the largest- and smallest-volume hyper-rectangles in the class $S^{(i_1,\ldots,i_d)}$.
     Any translated copy of $S$ that intersects $L$ must be contained within a hyper-rectangle $L'=\prod_{j=1}^d[a(x_j),b(x_j)]$ such that for each $j\in[d]$, we have $|b(x_j)-a(x_j)|<2^{i_j+2}$.
Therefore, the problem reduces to packing non-overlapping copies of $S$ inside
    $L'$. Since the ratio between the volume of $L'$ and the volume of $S$ is $4^d$, at most $4^d$ pairwise non-intersecting copies of $S$ can be placed inside $L'$. Hence, $\zeta^{(i_1,\ldots,i_d)}$ is at most~$4^d$. \myqed
\end{proof}

Now, using~\Cref{IKN_rectangle},  the analysis of $\RCY$ is analogous to that of $\ACY$ given in~\Cref{thm:ACY}. The main  difference is that algorithm $\RCY$ selects a class $S^{(i_1,\ldots,i_d)}$  based on independently chosen indices $i_1,i_2,\ldots,i_d$ chosen uniformly at random from $\{0,1,\ldots,\lfloor \log M\rfloor\}$. The probability of selecting a class $S^{(i_1,\ldots,i_d)}$ is $\left(\frac{1}{\lfloor \log M\rfloor+1}\right)^d$ instead of $\frac{1}{\lfloor \log M\rfloor+1}$.  Combining this selection probability with the deterministic guarantee of $\AGDS$ and the bound from~\Cref{IKN_rectangle}, we complete the proof of~\Cref{thm:RCY}.

\section{Conclusion}\label{sect_conclusion}

In this paper, we first establish optimal deterministic bounds for the online MIS problem in general graphs, parameterized by the independent kissing number. Building on this foundation, we showed that randomization, together with access to the geometric representation, can be used to surpass deterministic lower bounds beyond disk graphs with bounded radii. This demonstrates that geometric structure and randomization can be leveraged meaningfully in the online setting, even beyond highly constrained geometric classes. 

Our results give rise to several natural open problems:
\begin{itemize}
    \item Can randomization with access to geometric representation surpass deterministic lower bounds for unit balls in $\IR^d$ when $d\ge 4$?
    \item Can similar techniques be used to beat deterministic lower bounds for other geometric objects, such as unit hypercubes in $\IR^d$?
    \item For similarly sized $\alpha$-fat objects, is it possible to close the gap between the known randomized lower bound of $\Omega(\log M)$ and the current upper bound of $O(\zeta'\log M)$?
\end{itemize}

\bibliography{references}

\end{document}